\documentclass[letterpaper,10pt,twocolumn, conference]{ieeeconf}

\IEEEoverridecommandlockouts
\usepackage{amsmath}
\usepackage{amsthm}
\usepackage{amssymb}
\usepackage{amsfonts}
\usepackage{bm}
\usepackage{cite}
\usepackage{xcolor}
\usepackage{graphicx}
\usepackage{tikz}
\usepackage{pgfplots}
\pgfplotsset{compat=1.18}
\usepackage[mode=buildnew]{standalone}
\usepackage{subcaption}
\usepackage{scalefnt}

\newlength\fwidth
\newcommand{\R}[1]{\ensuremath{\mathbb{R}^{#1}}}
\renewcommand{\S}[1]{\ensuremath{\mathbb{S}_{#1}}}
\renewcommand{\P}[2][]{\mathbb{P}_{#1} ( #2 ) }

\newcommand{\cN}{\bm{\mathcal{N}}}

\newcommand{\Zplus}{\mathbb{Z}\sub{+}}

\newcommand{\E}[1]{\mathbb{E}\left[#1\right]}
\newcommand{\Cov}[1]{\mathrm{Cov}\left( #1 \right)}
\newcommand{\tr}[1]{\ensuremath{\mathrm{tr} \left( #1 \right) }}
\newcommand{\vertcat}[1]{\ensuremath{\mathrm{vertcat} (#1)}}
\newcommand{\bdiag}[1]{\ensuremath{\mathrm{bdiag} (#1)}}

\newcommand{\transpose}{\ensuremath{^{\mathrm{T}}}}
\newcommand{\sub}[1]{\ensuremath{_{#1}}}
\newcommand{\spr}[1]{\ensuremath{^{#1}}}

\newcommand{\curly}[1]{ \ensuremath{ \{ #1 \} } }

\newtheorem{theorem}{Theorem}

\newtheorem{proposition}{Proposition}
\newtheorem{problem}{Problem}
\theoremstyle{definition}
\newtheorem{remark}{Remark}

\newtheorem{assumption}{Assumption}

\renewenvironment{proof}
    {\noindent \textit{Proof.}}
    {\hfill \ensuremath{\blacksquare}}

\title{Density Steering of Gaussian Mixture Models for Discrete-Time Linear Systems}
\author{Isin M. Balci \and Efstathios Bakolas}
\date{September 2023}

\begin{document}

\maketitle

\begin{abstract}
In this paper, we study the finite-horizon optimal density steering problem for discrete-time stochastic linear dynamical systems. Specifically, we focus on steering probability densities represented as Gaussian mixture models which are 
known to give good approximations for general smooth probability density functions. We then revisit the covariance steering problem for Gaussian distributions and derive its optimal control policy. Subsequently, we propose a randomized policy to enhance the numerical tractability of the problem and demonstrate that under this policy the state distribution remains a Gaussian mixture. By leveraging these results, we reduce the Gaussian mixture steering problem to a linear program. 
We also discuss the problem of steering general distributions using Gaussian mixture approximations. 
Finally, we present results of non-trivial numerical experiments and demonstrate that our approach 
can be applied to general distribution steering problems.
\end{abstract}

\section{Introduction}\label{s:introduction}
In this paper, we consider the problem of characterizing control policies that can steer the probability distribution of the initial state of a discrete-time stochastic linear dynamical system to a target distribution in finite time. 
This type of problems can be classified as a variant of the ``optimal mass transport" problem \cite{p:chen2016OMToverLinearDS}. 
Specifically, we focus on Gaussian mixture models to represent the initial and target state probability distributions due to their universal approximation properties \cite[Chapter 3]{b:stergiopoulos2017advSignalProHandbook}.

Typically, three different approaches are employed to solve density steering problems. In the first approach, continuous-time dynamics are utilized to derive the Fokker-Planck partial differential equation (PDE), describing the evolution of the probability density function of the state. Control policies are designed for this PDE using Lyapunov-based methods \cite{p:eren2017velocityFieldDensityControl, p:zheng2022backsteppingDensityControl}. The second approach considers a discrete state space and employs Markov chain-based methods, and utilizes convex optimization tools to design a transition matrix of the Markov chain for steering the probability distribution \cite{p:demir2015decentralizedProbabilisticDensityControl, p:djeumou2022probabilisticSwarmGTL}. Lastly, optimal mass transport-based approaches treat the dynamic problem as a static mass transport problem using transition costs between initial and terminal states. 
This associated problem is solved using well-known optimal mass transport algorithms \cite{p:chen2016OMToverLinearDS, p:debadyn2021discreteTimeLQRviaOT, p:ito2023entropicMPCOT, p:krishnan2018distributedOTforSwarm}.

Our approach is primarily aligned with optimal mass transport-based methods. 
However, rather than seeking the optimal transport map in a general discretized state-space, we adopt Gaussian mixture models in conjunction with covariance steering theory \cite{p:chen2015optimalCSI, p:goldshtein2017finitehorizonCS, p:balci2022exact}. 
This enables us to formulate and solve a lower-dimensional linear program, offering a more efficient and effective approach.

\noindent \textit{Literature Review:} The problem of density steering has received significant attention in previous research. In \cite{p:debadyn2021discreteTimeLQRviaOT}, the authors propose a similar problem and re-frame it as an optimal mass transport problem. 
Meanwhile, in \cite{p:terpin2023DPinProbSpace}, the authors tackle the density steering problem using dynamic programming over probability spaces. However, it is worth noting that both of teh aforementioned approaches involve discretizing the continuous state-space, as the computation of the optimal transport map and Bellman recursion necessitates such discretization.

Additionally, Markov chain-based density steering methods also rely on discretizing the continuous state space to synthesize a Markov chain that characterizes the probabilistic motion of individual agents. In the infinite horizon case, formulating this Markov chain can be cast as a convex semi-definite program \cite{p:demir2015decentralizedProbabilisticDensityControl}, while in the finite horizon case the problem is not convex and should be addressed as a general nonlinear program \cite{p:hoshino2021MCWassersteinCarSharing, p:djeumou2022probabilisticSwarmGTL}.

In \cite{p:halder2021NonlinearDensity}, the authors derive the necessary conditions for optimality for the density steering problem with nonlinear drift in terms of coupled PDEs which they solve by using the Feynman-Kac lemma and point cloud sampling. In \cite{p:sivaramakrishnan2022distributionSteeringGeneralDist}, the authors study the distribution steering problem for linear systems excited by non-Gaussian noise using characteristic functions. In \cite{p:wu2023powerMomentSteering}, the authors employ power moments to formulate convex optimization problems for steering general probability densities in a one-dimensional setting.
A PDE-based optimal robotic swarm coverage control policy is obtained by deriving necessary conditions of optimality in \cite{p:elamvazhuthi2015optimalCoverageControl}.
For systems affected by multiplicative noise, infinite horizon optimal density steering laws have been derived in \cite{p:bakshi2020stabilizingDensityControl}.
Finally, in \cite{p:saravanos2023distributedLargeScale}, a hierarchical clustering-based density steering algorithm is presented for distributed large-scale applications.

The density steering problem using Gaussian mixture models has been investigated in the past \cite{p:abdulghafoor2023LargeScaleControl, p:zhu2021adaptive}.
However, a notable limitation of \cite{p:abdulghafoor2023LargeScaleControl, p:zhu2021adaptive} is that their proposed density steering methods do not explicitly account for the system's dynamics. Moreover, their approach relies on spatial discretization methods, thereby failing to fully exploit the advantages offered by Gaussian mixture models.

\noindent \textit{Main Contributions:} 
All of the aforementioned approaches for density steering problems require either state space discretization, solving partial differential equations, or extensive sampling. In contrast, our approach sidesteps these issues, offering a computationally efficient solution to the optimal density steering problem.
To this aim, we firstly derive a closed-form solution for the optimal state feedback policy, which is designed to guide an initial Gaussian state distribution to a desired terminal state while considering quadratic input and state costs. 
Secondly, we introduce a class of randomized state feedback policies aimed at simplifying the primary problem into a finite-dimensional optimization challenge. We also demonstrate that this control policy ensures that the state density will remain a Gaussian mixture model for the whole time horizon. Thirdly, we reveal that the intricate finite-dimensional optimization problem can be effectively reduced to a linear program. 
Finally, we showcase the effectiveness of our methodology in steering arbitrary densities through Gaussian mixture approximation, leveraging the expectation-maximization algorithm.

\section{Problem Formulation}\label{s:problem-formulation}
\noindent \textbf{Notation:} $\R{n}$ ($\R{n \times m}$) denotes the space of $n$-dimensional real vectors ($n\times m$ matrices). $\Zplus$ represents positive integers. 
Positive definite (semi-definite) $n\times n$ matrices are denoted by $\S{n}^{++}$ ($\S{n}^{+}$). 
For any $x\in \R{n}$ and $Q \in \S{n}^{+}$, $\lVert x \rVert_Q$ is the notation for $\sqrt{x\transpose Q x }$. 
The identity matrix of size $n \times n$ is $I\sub{n}$. 
Vertical concatenation of vectors or matrices $x_0, \dots, x_N$ is denoted as $\vertcat{x_0, x_1, \dots, x_N}$. 
The trace operator is represented by $\tr{\cdot}$. 
The nuclear norm of a matrix $A \in \R{n \times m}$ is denoted by $\lVert A \rVert_{*}$.
The block diagonal matrix with diagonal blocks $A_1, \dots, A_{N}$ is written as $\bdiag{A_1, \dots, A_{N}}$. 
Expectation and covariance of a random variable $x$ are given as $\E{x}$ and $\Cov{x}$, respectively. 
The notation $x \sim \cN (\mu, \Sigma)$ means that $x$ is a Gaussian random variable with mean $\mu$ and covariance $\Sigma$.
We use $\Delta_n$ to denote the probability simplex in $\mathbb{R}^n$, where $\Delta_n:=\{[p_1,\dots, p_n]\transpose\in\mathbb{R}^n: \sum_{i=1}^n p_i = 1 ~\text{and}~ p_i\geq 0~ \forall i\}$.
When $x$ follows a Gaussian mixture model, we write $x \sim \mathbf{GMM}(\curly{p\sub{i}, \mu_i, \Sigma_i}\sub{i=0}\spr{n-1})$ where $\curly{p_i}_{i=0}\spr{n-1} \in \Delta_n$. 
The probability density function of a random variable $x \in \R{n}$ evaluated at $x' \in \R{n}$ is $\P[x]{x'}$.
If $x\sim \mathcal{N}(\mu, \Sigma)$, the $\P[x]{x'}$ is denoted as $\P[\mathcal{N}]{x'; \mu, \Sigma}$.
For random variables $x \in \R{n}, y \in \R{m}$, $x|y=\Hat{y}$ denotes the conditional random variable $x$ given $y=\Hat{y}$. 
Finally, for an arbitrary set $\mathcal{A} \subseteq \R{n}$, $\mathcal{P} (\mathcal{A})$ represents the set of all random variables over $\mathcal{A}$.

\noindent \textbf{Problem Setup and Formulation:}
We consider a discrete-time stochastic linear dynamical system:
\begin{equation}\label{eq:linear-system-eq}
    x\sub{k+1} = A_k x_k + B_k u_k
\end{equation}
where $x_k \in \R{n}$ and $u\sub{k} \in \R{m}$ are the state and the input processes, respectively. 
We assume that $x\sub{0} \sim \mathbf{GMM}(\curly{p\sub{i}\spr{0}, \mu_{i}\spr{0}, \Sigma_{i}\spr{0}}\sub{i=0}\spr{r-1})$ such that $\curly{p_i}_{i=0}\spr{r-1} \in \Delta_r$, $\mu\sub{i}\spr{0} \in \R{n}$ and $\Sigma\sub{i}\spr{0} \in \mathbb{S}\sub{n}\spr{++}$ for all $i \in \curly{0, \dots, n-1}$.

A control policy with a horizon $N \in \mathbb{Z}\spr{+}$ for the system in \eqref{eq:linear-system-eq} is defined as a sequence of control laws $\pi = \curly{\pi\sub{i}}\sub{i=0}\spr{N-1}$ where each $\pi\sub{k} : \R{n} \rightarrow \mathcal{P} (\R{m})$ is a function that maps the state $x_k$ to a random variable representing control inputs. 
The set of randomized control policies is denoted by $\Pi$.
Throughout the paper, we consider quadratic cost functions in the form:
\begin{align}\label{eq:quadratic-cost-definitin}
    J(X\sub{0:N}, U\sub{0:N-1}) = & \sum\sub{k=0}\spr{N-1} \lVert u\sub{k} \rVert\sub{R\sub{k}}\spr{2} + \lVert x\sub{k} - x'\sub{k} \rVert\sub{Q\sub{k}}\spr{2} \nonumber \\
    & \qquad + \lVert x_N - x_N' \rVert_{Q_N}\spr{2}
\end{align}
where $X\sub{0:N} = \curly{x\sub{k}}\sub{k=0}\spr{N}$ is the state process and $U\sub{0:N-1} = \curly{u\sub{k}}\sub{k=0}\spr{N-1}$ is the input process.
Since we have defined the system dynamics and control policies, now we provide the formulation of the Gaussian mixture model (GMM) density steering problem:
\begin{problem}[GMM Density Steering Problem]\label{prob:main-problem-definition}
Let $N \in \mathbb{Z}\spr{+}$, $A\sub{k} \in \R{n \times n}, B\sub{k} \in \R{n \times m}$, $R\sub{k} \in \S{m}\spr{++}, Q\sub{k} \in \S{n}\spr{+}$  be given for all $k \in \curly{0, \dots, N-1}$. 
Also, let $r, t \in \mathbb{Z}\spr{+}$, $\curly{p\sub{i}\spr{0}}\sub{i=0}\spr{r-1} \in \Delta_r$, $\curly{p\sub{i}\spr{d}}\sub{i=0}\spr{t-1} \in \Delta_t$, $\curly{\mu\spr{0}\sub{i}}\sub{i=0}\spr{r-1}$, $\curly{\mu\spr{d}\sub{i}}\sub{i=0}\spr{t-1}$, 
$\curly{\Sigma\spr{0}\sub{i}}\sub{i=0}\spr{r-1}$, 
$\curly{\Sigma\spr{d}\sub{i}}\sub{i=0}\spr{t-1}$ such that $\mu\sub{i}\spr{0} , \mu\sub{i}\spr{d} \in \R{n}$ and $\Sigma\sub{i}\spr{0}, \Sigma\sub{i}\spr{d} \in \S{n}\spr{++}$ be given. 
Find an admissible control policy $\pi\spr{\star} \in \Pi$ that solves the following problem:
\begin{subequations}\label{eq:main-problem-opt-problem}
\begin{align}
    \min\sub{\pi \in \Pi} & ~~ \E{ J(X_{0:N}, U_{0:N-1}) } \\
    \text{s.t.} & ~~ \eqref{eq:linear-system-eq} \nonumber \\
    & ~~ x\sub{0} \sim \mathbf{GMM}\left(\curly{p\sub{i}\spr{0}, \mu\spr{0}\sub{i}, \Sigma\spr{0}\sub{i}}\sub{i=0}\spr{r-1} \right) \\
    & ~~ x\sub{N} \sim \mathbf{GMM}\left(\curly{p\sub{i}\spr{d}, \mu_i\spr{d}, \Sigma_i\spr{d}}\sub{i=0}\spr{t-1}\right) \\
    & ~~ u_k = \pi\sub{k}(x\sub{k})
\end{align}
\end{subequations}
\end{problem}

Throughout the paper, we make the following assumption regarding the controllability of the system in \eqref{eq:linear-system-eq}. 
Note that this is not a limiting assumption since many real-world linear dynamical systems are actually controllable. 
\begin{assumption}\label{assumption:linear-controllable}
    The system dynamics given in \eqref{eq:linear-system-eq} is controllable over a given problem horizon $N$. 
    In other words, the controllability grammian defined as:
    \begin{equation}\label{eq:ctrl-gram-def}
        \mathcal{G}_{N:0} = \sum_{k=0}^{N-1} \Phi_{N, k+1} B_k B_k\transpose \Phi_{N, k+1}\transpose  
    \end{equation}
    is non-singular with $\Phi_{k_2, k_1} := A_{{k_2}-1} A_{{k_2}-2} \dots A_{k_1}$, $\Phi_{k, k} = I_n$ for all $k_2, k_1 \in \mathbb{Z}^+$ such that $k_2 \geq k_1$.
\end{assumption}

\section{Optimal Covariance Steering for Gaussian Distributions}\label{s:optimal-cs-for-gaussian}
Optimal covariance steering problems for linear dynamical systems with quadratic cost functions have been extensively studied in the literature \cite{p:chen2015optimalCSI, p:goldshtein2017finitehorizonCS}. 
A closed-form solution to the optimal covariance steering problem is provided in \cite{p:goldshtein2017finitehorizonCS} for the case where the cost function is a convex quadratic function of the input. 
In this section, we expand upon these findings and derive a closed-form solution to the covariance steering problem, considering quadratic cost functions that depend on both the state and input. 
This result will be subsequently applied in Section \ref{s:reduce-to-lp} to formulate Problem \ref{prob:main-problem-definition} as a linear program. 
To begin, we revisit the formal definition of the Gaussian covariance steering problem:
\begin{problem}[Gaussian Covariance Steering]\label{prob:covariance-steering}
Let $\mu_0, \mu_d \in \R{n}$, $\Sigma\sub{0}, \Sigma\sub{d} \in \S{n}\spr{++}$, $R\sub{k} \in \S{m}\spr{++}$, $Q\sub{k} \in \S{n}\spr{+}$ for all $k \in \curly{0, \dots, N}$ be given. 
Find an admissible control policy $\pi\spr{\star} \in \Pi$ that solves the following problem:
\begin{subequations}\label{eq:GaussianCovSteer-Opt-Problem}
\begin{align}
    \min\sub{\pi \in \Pi} & ~~ \E{J(X_{0:N}, U_{0:N-1})} \\
    \text{s.t.} & ~~ \eqref{eq:linear-system-eq} \nonumber \\
    & ~~ x\sub{0} \sim \mathcal{N} (\mu_0, \Sigma_0) \\
    & ~~ x\sub{N} \sim \mathcal{N} (\mu_d, \Sigma_d) 
\end{align}
\end{subequations}
\end{problem}
In \cite{p:balci2022exact}, it is demonstrated that the optimal policy for the linear-quadratic Gaussian density steering problem takes the form of a deterministic affine state feedback policy, expressed as $\pi_k(x_k) = \Bar{u}\sub{k} + K_k (x_k - \mu_k)$ where $\mu_k = \E{x_k}$. 
Furthermore, for deterministic linear systems, this affine state feedback policy can be equivalently expressed in terms of the initial state as $\pi_k(x_0) = \Bar{u}_k + L_k (x_0 - \mu_0)$ \cite{p:goldshtein2017finitehorizonCS}. 
Consequently, the optimal Gaussian covariance steering problem can be rewritten equivalently in terms of decision variables $\curly{\Bar{u}_k, L_k}_{k=0}^{N-1}$ as follows:
\begin{subequations}\label{eq:Gaussian-Cov-Steer-Finite-Dim}
\begin{align}
    \min\sub{\mathbf{\Bar{U}}, \mathbf{L}} & ~~ \mathbf{\Bar{U}}\transpose \mathbf{R} \mathbf{\Bar{U}} + \tr{ \mathbf{R} \mathbf{L} \Sigma_{0} \mathbf{L}\transpose } + \mathbf{\Tilde{X}}\transpose \mathbf{Q} \mathbf{\Tilde{X}} \nonumber \\
    & ~~ + \tr{\mathbf{Q} (\mathbf{\Gamma} + \mathbf{H_u} \mathbf{L}) \Sigma_0 (\mathbf{\Gamma} + \mathbf{H_u} \mathbf{L})\transpose} \\
    \text{s.t.} & ~~ \mu_d = \Phi_{N, 0} \mu_0 + \mathbf{B_N} \mathbf{\Bar{U}},  \label{eq:Gaussian-CS-constr-mean}\\
    & ~~ \Sigma_d = (\Phi_{N, 0} + \mathbf{B_N} \mathbf{L}) \Sigma_0 (\Phi_{N, 0} + \mathbf{B_N} \mathbf{L})\transpose, \label{eq:Gaussian-CS-constr-cov}
\end{align}
\end{subequations}
where $\mathbf{X} := \vertcat{x_0, \dots, x_N}$, $\mathbf{\Bar{X}} := \E{\mathbf{X}}$, $\mathbf{\Tilde{X}} := \mathbf{\Bar{X}} - \mathbf{X'}$, , $\mathbf{X'} := \vertcat{x_0', \dots, x_N'}$, $\mathbf{U} := \vertcat{u_0, \dots, u_{N-1}}$, $\mathbf{\Bar{U}} := \E{\mathbf{U}} = \vertcat{\Bar{u}_0, \dots, \Bar{u}_{N-1}}$, $\mathbf{Q}:=\bdiag{Q_0, \dots, Q_{N}}$, $\mathbf{R} := \bdiag{R_0, \dots, R_{N-1}}$, $\mathbf{L}:=\vertcat{L_0, \dots, L_{N-1}}$, $\mathbf{B_N}:= [\Phi_{N, 1}B_0, \Phi_{N, 2}B_1, \dots, \Phi_{N, N}B_{N-1}]$. Note that $\mathbf{B_N} \mathbf{B}_{\mathbf{N}}\transpose = \mathcal{G}_{N:0}$ where $\mathcal{G}_{N:0}$ is defined as in \eqref{eq:ctrl-gram-def}.
The matrices $\mathbf{\Gamma}, \mathbf{H_u}$ is derived from the vertical concatenation of the state and input vectors as the concatenated vectors satisfy the following equalities:
\begin{subequations}\label{eq:concat-dynamics}
\begin{align}
    \mathbf{X}  & = \mathbf{\Gamma} x_0 + \mathbf{H_u U}, \\
    \mathbf{U} & = \mathbf{L} (x_0 - \mu_0)  + \mathbf{\Bar{U}}.
\end{align}
\end{subequations}
The derivation of \eqref{eq:concat-dynamics} and the expressions of the matrices $\mathbf{\Gamma}$ and $\mathbf{H_u}$ are omitted due to space restrictions and the readers can refer to \cite{p:goldshtein2017finitehorizonCS}.

The constraints in \eqref{eq:Gaussian-CS-constr-mean} and \eqref{eq:Gaussian-CS-constr-cov} correspond to the mean and covariance steering constraints in the problem defined in \eqref{eq:Gaussian-Cov-Steer-Finite-Dim}. 
Since the state mean depends on $\mathbf{\Bar{U}}$ and the state covariance depends on $\mathbf{L}$, and the objective function is separable in $\mathbf{\Bar{U}}$ and $\mathbf{L}$, the mean and covariance steering problems can be decoupled. 
The mean steering problem is formulated as follows:
\begin{align}\label{eq:Mean-Steering-Problem}
    \min_{\mathbf{\Bar{U}}} & ~~ J_{\mathrm{mean}}(\mathbf{\Bar{U}}) \\
    \text{s.t.} & ~~ \eqref{eq:Gaussian-CS-constr-mean} \nonumber
\end{align}
where 
\begin{align}\label{eq:mean-steering-objective}
    J_{\mathrm{mean}}(\mathbf{\Bar{U}}) := &   \mathbf{\Bar{U}}\transpose \mathbf{R} \mathbf{\Bar{U}} \nonumber \\
    & + (\mathbf{\Gamma}\mu_0 + \mathbf{H_u \Bar{U}})\transpose \mathbf{Q} \left( \boldsymbol{\Gamma}  \mu_0 + \mathbf{H_u} \mathbf{\Bar{U}} \right).
\end{align}
Note that the problem in \eqref{eq:Mean-Steering-Problem} is a strictly convex quadratic program with affine equality constraints since $R_k \in \S{m}\spr{++}$. The closed-form solution to such problems can be obtained using the first-order conditions of optimality (KKT conditions) \cite{b:boyd_vandenberghe_2004}. The following proposition provides the optimal feed-forward control input $\mathbf{\Bar{U}}$.
\begin{proposition}\label{prop:optimal-mean-steering}
Under Assumption \ref{assumption:linear-controllable}, the optimal control sequence $\mathbf{\Bar{U}}\spr{\star}$ that solves the problem in \eqref{eq:Mean-Steering-Problem} is given by:
\begin{align}
    \Lambda\spr{\star} = &  2 (\mathbf{B_N} M^{-1} \mathbf{B}_{\mathbf{N}}\transpose)^{-1} \times \nonumber \\
    & ~~ (\mathbf{B_N} M^{-1} \mathbf{H_u}\transpose \mathbf{Q} Y + (\mu_d - \Phi_{N,0} \mu_0 )), \\
    \mathbf{\Bar{U}}\spr{\star} = &  (1/2) M^{-1} (\mathbf{B_N}\transpose \Lambda\spr{\star} - 2 \mathbf{H_u}\transpose \mathbf{Q} Y),
\end{align}
where $M = \mathbf{R} + \mathbf{H_u Q H_u}\transpose$, $Y = \mathbf{\Gamma} \mu_0 - \mathbf{\Bar{X}}$.
\end{proposition}

Similarly, the decoupled covariance steering problem can be written as follows:
\begin{align}\label{eq:Cov-Steering-Problem}
    \min_{\mathbf{L}} & ~~ J_{\mathrm{cov}}(\mathbf{L}) \\
    \text{s.t.} & ~~ \eqref{eq:Gaussian-CS-constr-cov} \nonumber
\end{align}
where 
\begin{align}\label{eq:cov-steering-objective}
    J_{\mathrm{cov}}(\mathbf{L}) := & \tr{\mathbf{R} \mathbf{L}  \Sigma_0 \mathbf{L}\transpose} \nonumber \\
    & +\tr{ \mathbf{Q} (\mathbf{\Gamma} + \mathbf{H_u} \mathbf{L})\Sigma_0 (\mathbf{\Gamma} + \mathbf{H_u} \mathbf{L})\transpose }.
\end{align}
The objective function $J_{cov}(\mathbf{L})$ of the covariance steering problem given in \eqref{eq:Cov-Steering-Problem} is a convex quadratic function of the decision variable $\mathbf{L}$. 
However, the terminal covariance constraint in \eqref{eq:Gaussian-CS-constr-cov} is a non-convex quadratic equality constraint. The next proposition provides the closed-form solution to problem given in \eqref{eq:Cov-Steering-Problem} and the optimal value of the objective function as a function of the problem parameters.
\begin{proposition}\label{prop:optimal-covariance-steering}
Given that Assumption \ref{assumption:linear-controllable} holds for the system defined in \eqref{eq:linear-system-eq}, then optimal sequence of feedback controller gains $\mathbf{L}\spr{\star}$ that solves the problem in \eqref{eq:Cov-Steering-Problem} is given by:
\begin{subequations}
\begin{align}
    \mathbf{L}\spr{\star} & = \mathbf{h} + \mathbf{D} Z, \\
    \mathbf{h} & = \mathbf{B}_{\mathbf{N}}\transpose \mathcal{G}_{N:0}\spr{-1} \left( \Sigma_{d}^{1/2} T \Sigma_0^{-1/2} - \Phi_{N,0} \right), \\
    Z & = - \left( \mathbf{D}\transpose M \mathbf{D} \right)^{-1} \mathbf{D}\transpose \left( M \mathbf{h} + \mathbf{H}_{\mathbf{u}}\transpose \mathbf{Q} \mathbf{\Gamma} \right), \\
    T & = - V_{\Omega} U\transpose_{\Omega}, \\
    \Omega & = \Sigma_0^{1/2} \left( \Theta_5\transpose \mathbf{Q} \mathbf{H_u} - \Theta_4\transpose \mathbf{R} \right) \Theta_1 \mathbf{B_N}\transpose \mathcal{G}_{N:0}\spr{-1} \Sigma_{d}\spr{1/2},
\end{align}
\end{subequations}
where $\mathbf{D} \in \R{mN \times mN - n}$ is an arbitrary matrix whose columns are orthogonal to the vector space spanned by the columns of $\mathbf{B}_{\mathbf{N}}\transpose$, i.e., $\mathbf{B_N} \mathbf{D} = \bm{0}$, $M = \mathbf{R} + \mathbf{H}_\mathbf{u}\transpose \mathbf{Q} \mathbf{H_u}$, $\Omega = U_{\Omega} \Lambda_{\Omega} V_{\Omega}\transpose$ is the singular value decomposition of matrix $\Omega$, $\Theta_1 := I_{Nm} - \mathbf{D}(\mathbf{D}\transpose M \mathbf{D})^{-1} \mathbf{D}\transpose M$, $\Theta_2 := \mathbf{D}(\mathbf{D}\transpose M \mathbf{D})^{-1} \mathbf{D}\transpose \mathbf{H}_{\mathbf{u}}\transpose \mathbf{Q} \mathbf{\Gamma}$, $\Theta_3 := \mathbf{B}_{\mathbf{N}}\transpose \mathcal{G}_{N:0}\spr{-1} \Phi_{N:0}$, $\Theta_4:= \Theta_1 \Theta_3 + \Theta_2$, $\Theta_5 := \mathbf{\Gamma} - \mathbf{H_u} \Theta_4$.
Furthermore, the optimal value of the objective function is given as:
\begin{align}
    J_{\mathrm{cov}}\spr{\star} = & ~\tr{ \mathbf{R} \left( \Theta_1 \mathcal{Z} \Theta_1\transpose + \Theta_4 \Sigma_0 \Theta_4\transpose \right) } \nonumber \\ 
    & + \tr{ \mathbf{Q} \left( \mathbf{H_u} \Theta_1 \mathcal{Z} \Theta_1\transpose \mathbf{H}_\mathbf{u}\transpose + \Theta_5 \Sigma_0 \Theta_5\transpose \right) } \nonumber \\
    & - 2 \lVert \Omega \rVert_{*},
\end{align}
where $ \mathcal{Z} := \mathbf{B}_{\mathbf{N}}\transpose \mathcal{G}_{N:0}\spr{-1} \Sigma_d \mathcal{G}_{N:0}\spr{-1} \mathbf{B_N} $.
\end{proposition}
\begin{proof}
Observe that the constraint \eqref{eq:Gaussian-CS-constr-cov} can be equivalently written as follows:
\begin{subequations}
\begin{align}\label{eq:prop-opt-cov-steer-rewritten-constr-cov}
    T = \Sigma_d^{-1/2} (\Phi_{N:0} + \mathbf{B_N} \mathbf{L}) \Sigma_0^{1/2}, ~~  T T\transpose = I_{n}.
\end{align}
Also, pick a full-rank matrix $\mathbf{D} \in \R{mN \times mN - n}$ such that $\mathbf{B_N} \mathbf{D} = \bm{0}$. 
Using $\mathbf{D}$, we can write $\mathbf{L} = \mathbf{B}_{\mathbf{N}}\transpose Y + \mathbf{D} Z$ where $Y\in \R{n \times n}$ and $Z \in \R{Nm - n \times n}$. 
Note that there is a one to one mapping between $\mathbf{L}$ and the pair $Y, Z$ since both $\mathbf{B_N}$ and $\mathbf{D}$ are full rank and they have orthogonal columns. 
Thus, we can rewrite \eqref{eq:prop-opt-cov-steer-rewritten-constr-cov} as $T =  \Sigma_d^{-1/2} (\Phi_{N:0} + \mathcal{G}_{N:0}Y)\Sigma_0$. 
Thus, we obtain:
\begin{equation}\label{eq:LofTandZ}
    \mathbf{L} = \mathbf{B}_N\transpose \mathcal{G}_{N:0}\spr{-1} \left( \Sigma_d\spr{1/2} T \Sigma_0\spr{-1/2} - \Phi_{N:0}\right) + \mathbf{D} Z.
\end{equation}
By plugging the right-hand-side of \eqref{eq:LofTandZ} into $\mathbf{L}$ in $J_\mathrm{cov}(\mathbf{L})$, we can recast the optimization problem in \eqref{eq:Cov-Steering-Problem} as follows:
\begin{align}
    \min_{T, Z} & ~~ J_1(T, Z) \\
    \text{s.t.} & ~~ T T\transpose = I_{n}
\end{align}
where $J_1(T, Z) = J_{\mathrm{cov}} \big( \mathbf{B}_N\transpose \mathcal{G}_{N:0}\spr{-1} \big( \Sigma_d\spr{1/2} T \Sigma_0\spr{-1/2} - \Phi_{N:0} \big) + \mathbf{D} Z \big)$. 
Note that the objective function $J_1(T, Z)$ is a convex quadratic function. 
For a fixed $T$, $\min_{Z} J_1(T, Z)$ is an unconstrained convex quadratic program. 
Thus, the optimal $Z$ can be found by solving the equation $\nabla_Z J_1(T, Z) = \bm{0}$; it follows that 
\begin{align}
    Z^{\star}(T) = - (\mathbf{D}\transpose M \mathbf{D})^{-1} \mathbf{D}\transpose (M h(T) + \mathbf{H}_{\mathbf{u}}\transpose \mathbf{Q} \mathbf{\Gamma} )
\end{align}
where $h(T) = \mathbf{B_N} \mathcal{G}_{N:0}\spr{-1} (\Sigma_d\spr{1/2} T \Sigma_0^{-1/2} - \Phi_{N,0})$, $Z\spr{\star}(T)$ denotes the optimal $Z$ for a fixed $T$. 
By plugging the expression of $Z\spr{\star}(T)$ back into $J_{1}(T, Z)$, we obtain the following optimization problem:
\begin{align}
    \min_{T} & ~~ J_2(T) \\
    \text{s.t.} & ~~ T T\transpose = I_n 
\end{align}
where $J_2(T) = J_1(T, Z\spr{\star}(T))$.
Expanding $J_2(T)$, we obtain that $J_2(T) = C + 2 \tr{\Omega T} $ where the constant term $C = \tr{ \mathbf{R} \left( \Theta_1 \mathcal{Z} \Theta_1\transpose + \Theta_4 \Sigma_0 \Theta_4\transpose \right) } + \tr{ \mathbf{Q} \left( \mathbf{H_u} \Theta_1 \mathcal{Z} \Theta_1\transpose \mathbf{H}_\mathbf{u}\transpose + \Theta_5 \Sigma_0 \Theta_5\transpose \right) }$. 
Finally, from Von Neuman trace inequality \cite{p:mirsky1975vonNeumanTrace}, we obtain that $T^{\star} = \arg \min_{T T\transpose = I_n} \tr{\Omega T}$ is given as $- V_{\Omega} U_{\Omega}\transpose$ and $\tr{\Omega T\spr{\star}} = \tr{\Lambda_{\Omega}} = \sum_i \sigma_i(\Omega) = \lVert \Omega \rVert_{*}$. 
\end{subequations}
\end{proof}

\begin{remark}
Note that in the special case in which $\mathbf{Q} = 0$, the problem in \eqref{eq:Cov-Steering-Problem} is equivalent to the covariance steering problem that is studied in \cite{p:goldshtein2017finitehorizonCS}.
Thus, the optimal policy derived in Proposition \ref{prop:optimal-covariance-steering} is equal to the optimal policy defined in \cite[Eq. (20)]{p:goldshtein2017finitehorizonCS}.
\end{remark}

\section{Reduction to Linear Program}\label{s:reduce-to-lp}
To establish a computationally efficient problem formulation, we define a finite-dimensional set of randomized policies. 
This approach enables us to recast Problem \ref{prob:main-problem-definition} as a linear program. 
First, we characterize the set of policies that can transform the initial Gaussian mixture model corresponding to the initial state of the system into another (target) Gaussian mixture model. 
Then, we formulate a finite-dimensional nonlinear program using the Gaussian mixture steering policies. 
Next, we reduce the resulting nonlinear program into a linear program by applying the results presented in Proposition \ref{prop:optimal-mean-steering} and Proposition \ref{prop:optimal-covariance-steering}.

\subsection{GMM Steering Policies}\label{ss:gmm-steering-policies}
To solve Problem \ref{prob:main-problem-definition} efficiently, we propose a set of randomized state feedback policies $\Pi_{r} \subset \Pi$.
Every $\pi \in \Pi_r$ is a sequence of functions $\curly{\pi_{0}, \pi_{1} , \dots, \pi_{N-1} }$ such that each $\pi_{k} : \R{n} \rightarrow \mathcal{P} (\R{m})$ is given as follows:
\begin{align}\label{eq:gmm-policy-definition}
    \pi_{k} (x_0) =  L_k^{i, j} (x_0 - \Bar{\mu}^{i}_0) + \Bar{u}_k^{i, j} ~~ \text{w.p.} ~~ \gamma_{i, j}(x_0)
\end{align}
where $L_{k}\spr{i, j} \in \R{m \times n}$, $\Bar{u}_k\spr{i, j} \in \R{m}$ for all $k \in \curly{0, \dots, N-1}, i \in \curly{0, \dots, r-1}$ and $j \in \curly{0, \dots, t-1}$.
Furthermore, $\gamma_{i, j} : \R{n} \rightarrow \R{} $ is given by:
\begin{align}\label{eq:gamma-def}
\gamma_{i, j} (x_0) = \frac{p_i \P[\mathcal{N}]{x_0; \Bar{\mu}_0\spr{i}, \Bar{\Sigma}_0\spr{i}}}{\sum_{i=0}\spr{r-1} p_i \P[\mathcal{N}]{x_0 ; \Bar{\mu}_0^{i}, \Bar{\Sigma}_0\spr{i}}}\lambda_{i, j}
\end{align}
where $\sum_{j} \lambda_{i, j} = 1$ for all $i$ and $\lambda_{i, j}, \geq 0$.
Note that, $\sum_{i, j} \gamma_{i,j} = 1$ and $\gamma_{i, j} \geq 0$, and the policy $\pi$ will return a valid probability distribution over the control sequences for a given $x_0$.

The policy defined in \eqref{eq:gmm-policy-definition} is based on the intuition that under affine state feedback policies, the state distribution $x_k$ will remain Gaussian, provided that the initial state distribution is also Gaussian. 
The term $\frac{p_i \P[\mathcal{N}]{x_0; \Bar{\mu}_0\spr{i}, \Bar{\Sigma}0\spr{i}}}{\sum_{i=0}\spr{r-1} p_i \P[\mathcal{N}]{x_0 ; \Bar{\mu}_0^{i}, \Bar{\Sigma}_0\spr{i}}}$ represents the likelihood that the initial state $x_0$ will be drawn from the $i$th component of the Gaussian mixture model with components $( \curly{p_i, \Bar{\mu}_i, \Bar{\Sigma}_i}\sub{i=0}\spr{r-1})$.
Consequently, the control policy in \eqref{eq:gmm-policy-definition} is an ideal choice for steering probability distributions defined by Gaussian mixture models.
The following proposition states that if a policy $\pi \in \Pi_r$ is applied to the dynamical system in \eqref{eq:linear-system-eq} whose initial state is sampled from $\mathbf{GMM}(\curly{p_i, \mu_i, \Sigma_i}_{i=0}\spr{r-1})$, then the terminal state $x_N$ will have the Gaussian mixture distribution $\mathbf{GMM}(\curly{q_i, \mu_i\spr{f}, \Sigma_i\spr{f}}_{i=0}\spr{t-1})$ whose parameters are determined by the policy $\pi \in \Pi_r$.
\begin{proposition}\label{prop:gmm-steering-proposition}
Let $x_0 \in \R{n}$ be the initial state of the system given in \eqref{eq:linear-system-eq} such that $x_0 \sim \mathbf{GMM}(\curly{p_i, \mu_i\spr{0}, \Sigma_i\spr{0}}_{i=0}\spr{r-1})$ and $\pi \in \Pi_r$ with parameters $(\curly{\Bar{\mu}_i, \Bar{\Sigma}_i}_{i=0}\spr{r-1}, \curly{\lambda_{i, j}}_{i=0, j=0}\spr{r-1, t-1}, \curly{\Bar{u}\sub{k}\spr{i,j}, L\sub{k}\spr{i,j}}_{i=0, j=0, k=0}\spr{r-1, t-1, N-1})$ and
$\mu_i\spr{0} = \Bar{\mu}_{i}$, $\Sigma_i\spr{0} = \Bar{\Sigma}_i$ for all $i \in \curly{0, \dots, r-1}$. 
Furthermore, let $u_k = \pi_k(x_0)$ for all $k \in \curly{0, \dots, N-1}$, then $x_N \sim \mathbf{GMM}(\curly{q_j, \mu_j\spr{f}, \Sigma_j\spr{f}})$ such that 
\begin{subequations}
\begin{align}
    q_j & = \sum_{i=0}\spr{r-1} p_i \lambda_{i, j} \label{eq:prop-gmm-steer-eq1}, \\
    \mu\spr{f}_{j} & = \Phi_{N:0} \mu_{i}\spr{0} + \mathbf{B_N} \mathbf{\Bar{U}}\sub{i, j},  \label{eq:prop-gmm-steer-eq2} \\
    \Sigma\spr{f}_{j} & = (\Phi_{N:0} + \mathbf{B_N} \mathbf{L}\sub{i, j}) \Sigma_{i}\spr{0} (\Phi_{N:0} + \mathbf{B_N} \mathbf{L}\sub{i, j})\transpose, \label{eq:prop-gmm-steer-eq3}
\end{align}
\end{subequations}
where \eqref{eq:prop-gmm-steer-eq1}-\eqref{eq:prop-gmm-steer-eq3} hold  for all $j \in \curly{0, \dots, t-1}$, $\mathbf{L}_{i, j} = \vertcat{L_0\spr{i, j}, \dots L_{N-1}^{i,j}}$ and $\mathbf{\Bar{U}}_{i, j} = \vertcat{\Bar{u}_0\spr{i, j}, \dots, \Bar{u}_{N-1}\spr{i, j}}$.
\end{proposition}
\begin{proof}
First, note that, by virtue of Bayes' Theorem on conditional probability densities, the probability density function of $x_N$ can be written as follows:
\begin{align}\label{eq:pdf-x-n}
    \P[x_N]{\Hat{x}} & =  \int_{\R{n}} \int_{\R{m N}}  \P[x_N | x_0 = \Hat{x}_0, U = \Hat{U}]{\Hat{x}} ~ \P[U| x_0=\Hat{x}]{\Hat{U}}  \nonumber \\
    & \qquad \qquad \qquad \times\P[x_0]{\Hat{x}_0 }  \, \mathrm{d}\Hat{U} \, \mathrm{d}\Hat{x}_0 .
\end{align}
Furthermore, the conditional probability density functions in \eqref{eq:pdf-x-n} are given as follows:
\begin{align}
    & \P[x_N | x_0  = \Hat{x}_0, U=\Hat{U}]{\Hat{x}} = \delta (\Hat{x} = \Phi_{N:0} \Hat{x}_0 + \mathbf{B_N} \Hat{U}) \label{eq:conditional-xn-given-x0-u} \\
    & \qquad ~ \, \P[U | x_0 = \Hat{x}_0]{\Hat{U}} = \nonumber \\ 
    &  \quad \sum_{i=0}\spr{r-1} \sum_{j=0}\spr{t-1} \gamma_{i, j} (\Hat{x}_0) \delta ( \Hat{U} = \mathbf{L}_{i, j} (\Hat{x}_0 - \mu_i\spr{0}) + \mathbf{\Bar{U}}\sub{i, j} )  \label{eq:conditional-U-given-x0}
\end{align}
where $\delta(\cdot)$ denotes the dirac delta function, \eqref{eq:conditional-xn-given-x0-u} can be obtained since system dynamics in \eqref{eq:linear-system-eq} and \eqref{eq:conditional-U-given-x0} can be obtained from the definition of the policy set $\Pi_r$ in \eqref{eq:gmm-policy-definition}. 
Let us analyze the inner integral in \eqref{eq:pdf-x-n} first. 
Observe that $\P[x_0]{\Hat{x}_0}$ does not depend on $\Hat{U}$ and thus it can be taken out of the inner integral. 
Plug the expressions in \eqref{eq:conditional-xn-given-x0-u} and \eqref{eq:conditional-U-given-x0} into \eqref{eq:pdf-x-n} to obtain:
\begin{align}
    & = \int_{\R{mN}} \delta (\Hat{x} = \Phi_{N:0} \Hat{x}_0 + \mathbf{B_N}\Hat{U}) \times \nonumber \\ 
    & \quad \sum_{i=0, j=0}\spr{r-1, t-1} \gamma_{i, j} (\Hat{x}_0) \delta ( \Hat{U} = \mathbf{L}_{i, j} (\Hat{x}_0 - \mu_i\spr{0}) + \mathbf{\Bar{U}}\sub{i, j} ) \, \mathrm{d}\Hat{U}, \\
    & = \sum_{i=0, j=0}\spr{r-1, t-1} \gamma_{i, j} (\Hat{x}_0) \delta \big(\Hat{x} = \Phi_{N:0} \Hat{x}_0 \nonumber \\
    & \qquad \qquad \qquad \qquad + \mathbf{B_N} (\mathbf{L}_{i,j} (\Hat{x}_0 - \mu_i\spr{0}) + \mathbf{\Bar{U}}_{i, j}) \big). \label{eq:inner-integral-resolved}
\end{align}
Equation~\eqref{eq:inner-integral-resolved} is obtained by using the linearity of the integral operator and the properties of the dirac delta function. The expression in \eqref{eq:inner-integral-resolved} can be rewritten as $\sum_{i, j} \gamma_{i, j}(\Hat{x}_0) \delta (\mathbf{H}_{i, j} \Hat{x}_0 - h_{i, j} = \Hat{x})$ for brevity, where $\mathbf{H}_{i, j} := \Phi_{N:0} + \mathbf{B_N} \mathbf{L}_{i, j}$ and $h_{i, j} :=  \mathbf{B_N} (\mathbf{L}_{i, j} \mu_{i}\spr{0} - \mathbf{\Bar{U}}_{i, j})$. Observe that the denominator of $\gamma_{i, j}(\Hat{x}_0)$ defined in \eqref{eq:gamma-def} is equal to $\P[x_0]{\Hat{x}_0}$. By using this fact, it follows that 
\begin{align}\label{eq:prop-gmm-proof-PxN}
    & \P[x_N]{\Hat{x}} = \sum_{i=0}\spr{r-1} \sum_{j=0}\spr{t-1} p_i \lambda_{i, j} g_{i, j} (\Hat{x}),
\end{align}
where $g_{i, j} (\Hat{x}) := \int_{\R{n}} \mathcal{N}(\Hat{x}_0; \Bar{\mu}_{0}\spr{i}, \Bar{\Sigma}_{0}\spr{i}) \delta(\mathbf{H}_{i, j} \Hat{x}_0 - h_{i, j} = \Hat{x})  \, d\Hat{x}_0$. 
Furthermore,  
\begin{align}
    \hspace{-0.15cm} g_{i, j} (\Hat{x}) & = \int_{\R{n}} \P[\mathcal{N}]{\mathbf{H}_{i, j}^{-1}(z_{i, j} + h_{i, j}) ; \Bar{\mu}_0\spr{i}, \Bar{\Sigma}_0\spr{i}} \nonumber \\
    & \qquad  \times \lvert \mathbf{H}_{i, j}\spr{-1} \rvert \delta (z_{i,j} = \Hat{x}) \, d z_{i, j}, \label{eq:prop-gmm-proof-gij-1}\\
    & = \P[\mathcal{N}]{\mathbf{H}_{i,j}\spr{-1} (\Hat{x} + h_{i,j});\Bar{\mu}_0\spr{i}, \Bar{\Sigma}_{0}\spr{i}}  |\mathbf{H}_{i, j}^{-1}|, \label{eq:prop-gmm-proof-gij-2}\\
    & = \P[\mathcal{N}]{\Hat{x}; \mathbf{H}_{i,j} \mu_{i}\spr{0} - h_{i,j}, \mathbf{H}_{i,j} \Sigma_i\spr{0} \mathbf{H}_{i,j}\transpose} , \label{eq:prop-gmm-proof-gij-3}
\end{align}
where $\lvert \mathbf{H}_{i,j} \rvert$ is the determinant of $\mathbf{H}_{i, j}$.
Equation \eqref{eq:prop-gmm-proof-gij-1} is obtained by applying the variable transformation $z_{i, j} = \mathbf{H}_{i,j} \Hat{x}_0 - h_{i, j} $.
Then, the standard property of the dirac delta function yields \eqref{eq:prop-gmm-proof-gij-2}.
Expanding $\mathcal{N}(\mathbf{H}_{i,j}\spr{-1} (\Hat{x} + h_{i,j});\Bar{\mu}_0\spr{i}, \Bar{\Sigma}_{0}\spr{i} )$, we obtain \eqref{eq:prop-gmm-proof-gij-3}.
After that, we expand $\mathbf{H}_{i,j}$, $h_{i,j}$ and define \eqref{eq:prop-gmm-steer-eq2} and \eqref{eq:prop-gmm-steer-eq3} to obtain $g_{i,j}(\Hat{x}) := \mathcal{N}(\Hat{x}; \mu_j\spr{f}, \Sigma_j\spr{f})$.
Finally, the proof is concluded by plugging this expression into \eqref{eq:prop-gmm-proof-PxN} and defining \eqref{eq:prop-gmm-steer-eq1}.
\end{proof}

\begin{remark}
In the statement of Proposition \ref{prop:gmm-steering-proposition}, it is given that for each $j \in \curly{0, \dots, t-1}$ every $\mathbf{\Bar{U}}_{i, j}$ should satisfy \eqref{eq:prop-gmm-steer-eq2} and every $\mathbf{L}_{i, j}$ should satisfy \eqref{eq:prop-gmm-steer-eq3}. 
This seems like an extra condition for the policy to satisfy and limits the applicability of the policy proposed in \eqref{eq:gmm-policy-definition}. 
However, the number of terminal Gaussian components $t$ is actually determined by the parameters ${\mathbf{\Bar{U}}_{i, j}, \mathbf{L}_{i,j}}$.
To see this, take an arbitrary set $\curly{\mathbf{\Bar{U}}_\ell, \mathbf{L}_\ell }_{\ell=0}\spr{s}$ for each $i \in \curly{0, \dots, r-1}$ and $\ell \in \curly{0, \dots, s}$, define $\mu_{i\times s + \ell}\spr{f} = \Phi_{N:0} \mu_i\spr{0} + \mathbf{B_N} \mathbf{\Bar{U}}_{\ell}$ and $\Sigma_{i\times s + \ell}\spr{f} = (\Phi_{N:0} + \mathbf{B_N} \mathbf{L}_{\ell}) \Sigma_{i}\spr{0} (\Phi_{N:0} + \mathbf{B_N} \mathbf{L}_{\ell})\transpose$. 
Thus, we can set $t = s \times r$ to obtain \eqref{eq:prop-gmm-steer-eq2} and \eqref{eq:prop-gmm-steer-eq3}.
\end{remark}
\subsection{Reduction to LP}\label{ss:reduction-lp}
Since we have demonstrated that, within the set of policies $\Pi_r$ defined in \eqref{eq:gmm-policy-definition}, the initial Gaussian mixture state distribution is transformed into another Gaussian mixture, and these policies in $\Pi_r$ are parameterized by a finite number of decision variables, we can utilize the set of policies $\Pi_r$ to formulate a finite-dimensional optimization problem aimed at solving Problem \ref{prob:main-problem-definition}.

Minimizing Problem \ref{prob:main-problem-definition} over randomized policies defined in \eqref{eq:gmm-policy-definition} is written as a finite-dimensional nonlinear program in terms of decision variables $\curly{\lambda_{i,j}, \mathbf{\Bar{U}}_{i,j}, \mathbf{L}_{i,j}}_{i=0, j=0}\spr{r-1, t-1}$ as follows:
\begin{subequations}\label{eq:nlp-main-problem}
\begin{align}
    \min_{ \lambda_{i,j}, \mathbf{\Bar{U}}_{i,j}, \mathbf{L}_{i,j}} & ~ \mathcal{J}_{r} (\curly{\lambda_{i,j}, \mathbf{\Bar{U}}_{i,j}, \mathbf{L}_{i,j}}_{i=0, j=0}\spr{r-1, t-1}) \\
    \text{s.t.} & ~ \lambda_{i,j} \geq 0 \label{eq:nlp-main-problem-constr-0}\\
    & ~ \sum_{j=0}\spr{t-1} \lambda_{i,j} = 1,  \label{eq:nlp-main-problem-constr-1}\\
    & ~ \sum_{i=0}\spr{r-1} p_i\spr{0} \lambda_{i,j} = p_j\spr{d}, ~\,\, \forall j \in \curly{0, \dots, t-1} \label{eq:nlp-main-problem-constr-2}\\
    & ~ \mu_j\spr{d} = \Phi_{N:0} \mu_{i}\spr{0} + \mathbf{B_N} \mathbf{\Bar{U}}\sub{i, j} \label{eq:nlp-main-problem-constr-3}\\
    & ~ \Sigma_j\spr{d} = \mathbf{H}_{i,j} \Sigma_{i}\spr{0} \mathbf{H}_{i,j}\transpose \label{eq:nlp-main-problem-constr-4}
\end{align}
\end{subequations}
where $\mathbf{H}_{i,j} = \Phi_{N:0} + \mathbf{B_N} \mathbf{L}_{i,j}$, $\mathcal{J}_{r}(\curly{\lambda_{i,j}, \mathbf{\Bar{U}}_{i,j}, \mathbf{L}_{i,j}}) = \E{J(X_{0:N}, U_{0:N-1})}$.
Constraints in \eqref{eq:nlp-main-problem-constr-0} and \eqref{eq:nlp-main-problem-constr-1} are due to the parametrization of the control policy in \eqref{eq:gmm-policy-definition}. The constraints in \eqref{eq:nlp-main-problem-constr-2}, \eqref{eq:nlp-main-problem-constr-3} and \eqref{eq:nlp-main-problem-constr-4} are obtained by making the right hand side of the equalities in \eqref{eq:prop-gmm-steer-eq1}, \eqref{eq:prop-gmm-steer-eq2} and \eqref{eq:prop-gmm-steer-eq3} equal to $p_j\spr{d}$, $\mu\sub{j}\spr{d}$ and $\Sigma_j\spr{d}$, respectively.
Furthermore, constraints \eqref{eq:nlp-main-problem-constr-1} and \eqref{eq:nlp-main-problem-constr-2} are enforced for all $i \in \curly{0, \dots, r-1}$ and for all $j \in \curly{0, \dots, t-1}$, respectively. 
And the constraints \eqref{eq:nlp-main-problem-constr-0}, \eqref{eq:nlp-main-problem-constr-3} and \eqref{eq:nlp-main-problem-constr-4} are enforced for all $i, j$ in $\curly{0, \dots, r-1} \times \curly{0, \dots, t-1}$.
Objective function $\mathcal{J}_r (\cdot)$ can be rewritten in terms of the decision variables $\curly{\lambda_{i,j}, \mathbf{\Bar{U}}_{i,j}, \mathbf{L}_{i,j}}$ using the law of iterated expectations as:
\begin{align}\label{eq:nlp-main-problem-objective}
    \mathcal{J}_r (\cdot) := \sum_i \sum_j \lambda_{i,j} (J_\mathrm{mean}(\mathbf{\Bar{U}}_{i,j}) + J_\mathrm{cov}(\mathbf{L}_{i,j}))
\end{align}
where $J_\mathrm{mean}(\cdot)$ and $J_\mathrm{cov}(\cdot)$ are defined in \eqref{eq:mean-steering-objective} and \eqref{eq:cov-steering-objective}, respectively.
Note that the parameters $\curly{\Bar{\mu}_{i}\spr{0}, \Bar{\Sigma}_{i}\spr{0}}_{i=0}\spr{r-1}$ are also decision variables for the randomized policy in \eqref{eq:gmm-policy-definition}. 
However, we make them constant and equal to $\curly{\mu_i\spr{0}, \Sigma_i\spr{0}}$ in order to invoke Proposition \ref{prop:gmm-steering-proposition} and formulate the nonlinear program in \eqref{eq:nlp-main-problem}.

The non-convexity of the nonlinear program in \eqref{eq:nlp-main-problem} is due to the objective function $\mathcal{J}_r (\cdot)$ and the non-convex equality constraint \eqref{eq:nlp-main-problem-constr-4}. 
We observe that when $\lambda{i,j}$ is fixed, the objective function in \eqref{eq:nlp-main-problem-objective} becomes separable for each $(i,j)$ pair. 
These separated optimization problems for all $(i,j)$ are linear Gaussian covariance steering problems, and optimal policies for each can be found by invoking Proposition \ref{prop:optimal-mean-steering} and Proposition \ref{prop:optimal-covariance-steering}. 
The following theorem summarizes the main result of this section and describes how the optimal policy can be extracted from the solution to the linear program in \eqref{eq:main-lin-prog}.
\begin{theorem}\label{theorem:main-theorem}
Optimal parameters of the policy $\pi \in \Pi_r$ given in \eqref{eq:gmm-policy-definition} that solves Problem \ref{prob:main-problem-definition} can be obtained by solving the following linear program:
\begin{subequations}\label{eq:main-lin-prog}
\begin{align}
    \min_{\boldsymbol{\Tilde{\lambda}} \in \R{r\times t}} & ~~ \tr{ \mathbf{C}\transpose \boldsymbol{\Tilde{\lambda}} }  \label{eq:main-lin-prog-obj}\\
    \text{s.t.} & ~~ \boldsymbol{\Tilde{\lambda}} \mathbf{1}_t = \mathbf{p}_0, ~ \boldsymbol{\Tilde{\lambda}}\transpose \mathbf{1}_r = \mathbf{p}_d, ~ \boldsymbol{\Tilde{\lambda}} \geq 0, \label{eq:main-lin-prog-constr}
\end{align}
\end{subequations}
where $\mathbf{1}_{n}$ denotes the $n$-dimensional column vector whose entries are all equal to $1$ and the constraint $\boldsymbol{\Tilde{\lambda}} \geq 0$ is to be understood in the element-wise sense. In addition, $\mathbf{C}_{i,j} := \mathcal{J}\spr{\star}_{\mathrm{mean}}(\mu\spr{0}_{i}, \mu_{j}\spr{d}) + J\spr{\star}_{\mathrm{cov}}(\Sigma_i\spr{0}, \Sigma\spr{d}_{j} )$, where $\mathbf{C}_{i,j}$ is the entry at the $i$th row and the $j$th column of $\mathbf{C}$, $J_{\mathrm{mean}}\spr{\star}(\mu_i\spr{0}, \mu_j\spr{d})$ and $J_{\mathrm{cov}}\spr{\star}(\Sigma_i\spr{0}, \Sigma_j\spr{d})$ denote optimal mean and covariance steering costs between $x_0 \sim \mathcal{N}(\mu_i\spr{0}, \Sigma_i\spr{0})$ and $x_N \sim \mathcal{N} (\mu_j\spr{d}, \Sigma_j\spr{d})$. 
Furthermore, the optimal mixing weights $\lambda_{i,j}\spr{\star}$ of the policy $\pi \in \Pi_r$ are given by:
$\lambda_{i,j}\spr{\star} = \frac{\Tilde{\lambda}_{i,j}}{p_i\spr{0}}$,
where $\Tilde{\lambda}_{i,j}\spr{\star}$ is the optimal solution of the linear program in \eqref{eq:main-lin-prog}.
\end{theorem}
The formal proof of Theorem \ref{theorem:main-theorem} is omitted since the reduction of Problem \ref{prob:main-problem-definition} to the linear program defined in \eqref{eq:main-lin-prog} has already been described in detail in Section \ref{ss:reduction-lp}.

It is worth mentioning that While our approach is designed for steering probability distributions defined as Gaussian mixture models over linear systems, it can also be extended to steer general probability distributions, leveraging the universal approximation property of GMMs \cite{b:stergiopoulos2017advSignalProHandbook}. 
To approximate general distributions as GMMs, one can employ the expectation-maximization algorithm \cite{p:expectationMaximization} since these class of algorithms are efficiently implemented in open-source packages such as \cite{p:scikit-learn}.

\section{Numerical Experiments}\label{s:numerical-experiments}
In this section, we present the results of our numerical experiments. 
All numerical computations were conducted on a Mac M1 with 8GB of RAM. 
The first numerical example in Section \ref{ss:num-ex-gmm} is similar to the one solved in \cite{p:debadyn2021discreteTimeLQRviaOT}, for comparison purposes.
The second example in Section \ref{ss:num-exp-integrator} is provided to illustrate that our approach can be used to steer general probability distributions. 

\subsection{2-D LQR with GMM Distributions}\label{ss:num-ex-gmm}
The parameters of the system dynamics and the cost function are taken from Example V-A in \cite{p:debadyn2021discreteTimeLQRviaOT}.
In this example, $x_k \in \R{2}$ and $u_k \in \R{1}$ for all $k \in \curly{0, \dots, N}$ where $N = 10$.
\begin{align*}
    A_k = 
    \begin{bmatrix} 
    0.9 & -0.1 \\
    -0.1 & 0.8
    \end{bmatrix}, 
    \qquad
    B_k = 
    \begin{bmatrix}
        1 \\
        0
    \end{bmatrix}.
\end{align*}
Furthermore, $Q_k = I_2$, $x_k' = [0.0, 0.0]\transpose$ and $R_k = 1$. Also, initial and desired state distributions are given as $x_0 \sim \mathbf{GMM}(\curly{p\spr{d}_i, \mu_i\spr{0}, \Sigma_i\spr{0}}_{i=0}\spr{1})$ and $x_N \sim \mathbf{GMM}(\curly{p_i\spr{d}, \mu_i\spr{d}, \Sigma_i\spr{d}}_{i=0}\spr{1})$ where $p_0\spr{0} = 0.8$, $p_1\spr{0} = 0.2$, $\mu_0\spr{0} = [-0.5, -0.6]\transpose$, $\mu_1\spr{0} = [0.0, 0.0]\transpose$, $\Sigma_0\spr{0} = 0.02 I_2$ , $\Sigma_1\spr{0} = \left[ \begin{smallmatrix} 0.02 & 0.0 \\ 0.0 & 0.04 \end{smallmatrix} \right]$, $p_0\spr{d} =0.5$, $p_{0}\spr{d} =0.5$, $\mu_{0}\spr{d} = [0.5, 0.5]\transpose$, $\mu_1\spr{d}= [0.6, -0.6]\transpose$, $\Sigma_0\spr{d} = 0.02 I_2 $, $\Sigma_1\spr{d} = \left[ \begin{smallmatrix} 0.02 & 0.0 \\ 0.0 & 0.01 \end{smallmatrix} \right]$.
In Figure \ref{fig:initial-desired-distributions}, we illustrate the initial and desired distributions.  
\begin{figure}[ht]
    \centering
    \begin{subfigure}{0.49\linewidth}
\begin{tikzpicture}

\definecolor{darkgray176}{RGB}{176,176,176}

\begin{axis}[
tick align=outside,
tick pos=left,
x grid style={darkgray176},
xmin=-1.00250626566416, xmax=1.00250626566416,
xtick style={color=black},
xtick={-1,-0.5,0,0.5,1},
xticklabels={
  \(\displaystyle {\ensuremath{-}1}\),
  \(\displaystyle {\ensuremath{-}0.5}\),
  \(\displaystyle {0}\),
  \(\displaystyle {0.5}\),
  \(\displaystyle {1}\),
},
y grid style={darkgray176},
ymin=-1.00250626566416, ymax=1.00250626566416,
ytick style={color=black},
ytick={-1,-0.5,0,0.5,1},
yticklabels={
  \(\displaystyle {\ensuremath{-}1}\),
  \(\displaystyle {\ensuremath{-}0.5}\),
  \(\displaystyle {0}\),
  \(\displaystyle {0.5}\),
  \(\displaystyle {1}\)
}, 
xlabel=$x^1$, 
ylabel=$x^2$,
ylabel style={at={(-0.13,0.5)}},
label style={font=\LARGE}, 
tick label style={font=\Large}
]
\addplot graphics [includegraphics cmd=\pgfimage,xmin=-1.00250626566416, xmax=1.00250626566416, ymin=-1.00250626566416, ymax=1.00250626566416] {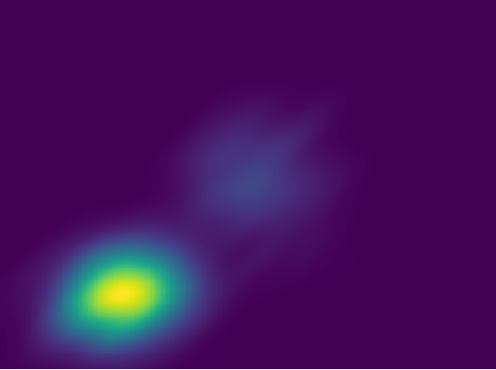};
\end{axis}

\end{tikzpicture}
    \subcaption{Initial}
    \end{subfigure}
    \begin{subfigure}{0.49\linewidth}
\begin{tikzpicture}

\definecolor{darkgray176}{RGB}{176,176,176}

\begin{axis}[
tick align=outside,
tick pos=left,
x grid style={darkgray176},
xmin=-1.00250626566416, xmax=1.00250626566416,
xtick style={color=black},
xtick={-1, -0.5, 0, 0.5 ,1},
xticklabels={
  \(\displaystyle {\ensuremath{-}1}\),
  \(\displaystyle {\ensuremath{-}0.5}\),
  \(\displaystyle {0}\),
  \(\displaystyle {0.5}\),
  \(\displaystyle {1}\),
},
y grid style={darkgray176},
ymin=-1.00250626566416, ymax=1.00250626566416,
ytick style={color=black},
ytick={-1, -0.5, 0, 0.5, 1},
yticklabels={
  \(\displaystyle {\ensuremath{-}1}\),
  \(\displaystyle {\ensuremath{-}0.5}\),
  \(\displaystyle {0}\),
  \(\displaystyle {0.5}\),
  \(\displaystyle {1}\),
},
xlabel=$x^1$, 
ylabel=$x^2$,
ylabel style={at={(-0.13,0.5)}},
label style={font=\LARGE}, 
tick label style={font=\Large}
]
\addplot graphics [includegraphics cmd=\pgfimage,xmin=-1.00250626566416, xmax=1.00250626566416, ymin=-1.00250626566416, ymax=1.00250626566416] {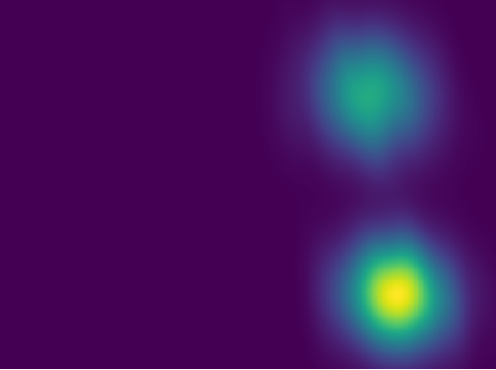};
\end{axis}

\end{tikzpicture}
    \subcaption{Desired}
    \end{subfigure}
    \caption{Initial and desired state distributions}
    \label{fig:initial-desired-distributions}
\end{figure}

In our experiment, the optimal policy required only 0.13 seconds for computation. 
This is notably faster than the reported computation time of 12 seconds for the same problem parameters in \cite{p:debadyn2021discreteTimeLQRviaOT}. 
The optimal $\boldsymbol{\lambda\spr{\star}} = \left[ \begin{smallmatrix}
0.25 & 0.75 \\ 1.0 & 0.0
\end{smallmatrix} \right]$, indicating that 3/4 of the probability mass in the first component of the initial probability distribution is transferred to the second component of the desired distribution. 
To illustrate the evolution of the densities, we collected 1000 samples from the initial state distribution and applied the optimal policy to these samples. 
The evolution of the probability density of the state over time is depicted in Figure \ref{fig:density-evolution}.
\begin{figure}
    \centering
    \begin{subfigure}{0.49\linewidth}
\begin{tikzpicture}

\definecolor{darkgray176}{RGB}{176,176,176}

\begin{axis}[
tick align=outside,
tick pos=left,
x grid style={darkgray176},
xlabel={$x^{1}$},
xmin=-2, xmax=2,
xtick style={color=black},
xtick={-2,-1,0, 1, 2},
xticklabels={
  \(\displaystyle {\ensuremath{-}2}\),
  \(\displaystyle {\ensuremath{-}1}\),
  \(\displaystyle {0}\),
  \(\displaystyle {1}\),
  \(\displaystyle {2}\)
},
y grid style={darkgray176},
ylabel={$x^2$},
ymin=-2, ymax=2,
ytick style={color=black},
ytick={-2,-1,0,1,2},
yticklabels={
  \(\displaystyle {\ensuremath{-}2}\),
  \(\displaystyle {\ensuremath{-}1}\),
  \(\displaystyle {0}\),
  \(\displaystyle {1}\),
  \(\displaystyle {2}\)
},
ylabel style={at={(-0.13,0.5)}},
label style={font=\LARGE}, 
tick label style={font=\Large}
]
\addplot graphics [includegraphics cmd=\pgfimage,xmin=-2, xmax=2, ymin=-2, ymax=2] {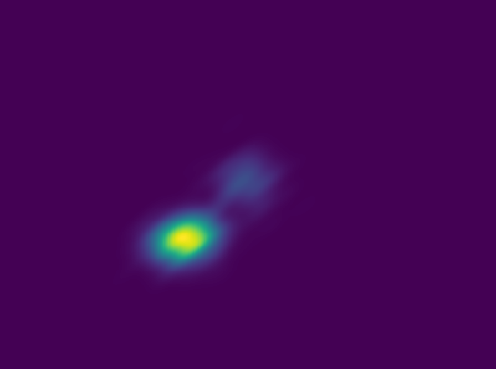};
\end{axis}
\end{tikzpicture}
    \subcaption{$k = 0$}
    \end{subfigure}
    \hfill
    \begin{subfigure}{0.49\linewidth}
\begin{tikzpicture}

\definecolor{darkgray176}{RGB}{176,176,176}

\begin{axis}[
tick align=outside,
tick pos=left,
x grid style={darkgray176},
xlabel={$x^1$},
xmin=-2, xmax=2,
xtick style={color=black},
xtick={-2,-1,0, 1, 2},
xticklabels={
  \(\displaystyle {\ensuremath{-}2}\),
  \(\displaystyle {\ensuremath{-}1}\),
  \(\displaystyle {0}\),
  \(\displaystyle {1}\),
  \(\displaystyle {2}\)
},
y grid style={darkgray176},
ylabel={$x^2$},
ymin=-2, ymax=2,
ytick style={color=black},
ytick={-2,-1,0, 1, 2},
yticklabels={
  \(\displaystyle {\ensuremath{-}2}\),
  \(\displaystyle {\ensuremath{-}1}\),
  \(\displaystyle {0}\),
  \(\displaystyle {1}\),
  \(\displaystyle {2}\)
},
ylabel style={at={(-0.13,0.5)}},
label style={font=\LARGE}, 
tick label style={font=\Large}
]
\addplot graphics [includegraphics cmd=\pgfimage,xmin=-2, xmax=2, ymin=-2, ymax=2] {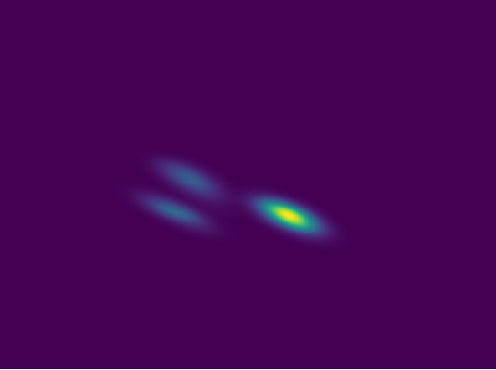};
\end{axis}

\end{tikzpicture}
    \subcaption{$k = 2$}
    \end{subfigure}
    \begin{subfigure}{0.49\linewidth}
\begin{tikzpicture}

\definecolor{darkgray176}{RGB}{176,176,176}

\begin{axis}[
tick align=outside,
tick pos=left,
x grid style={darkgray176},
xlabel={$x^1$},
xmin=-2, xmax=2,
xtick style={color=black},
xtick={-2,-1,0, 1, 2},
xticklabels={
  \(\displaystyle {\ensuremath{-}2}\),
  \(\displaystyle {\ensuremath{-}1}\),
  \(\displaystyle {0}\),
  \(\displaystyle {1}\),
  \(\displaystyle {2}\)
},
y grid style={darkgray176},
ylabel={$x^2$},
ymin=-2, ymax=2,
ytick style={color=black},
ytick={-2,-1,0, 1, 2},
yticklabels={
  \(\displaystyle {\ensuremath{-}2}\),
  \(\displaystyle {\ensuremath{-}1}\),
  \(\displaystyle {0}\),
  \(\displaystyle {1}\),
  \(\displaystyle {2}\)
},
ylabel style={at={(-0.13,0.5)}},
label style={font=\LARGE}, 
tick label style={font=\Large}
]
\addplot graphics [includegraphics cmd=\pgfimage,xmin=-2, xmax=2, ymin=-2, ymax=2] {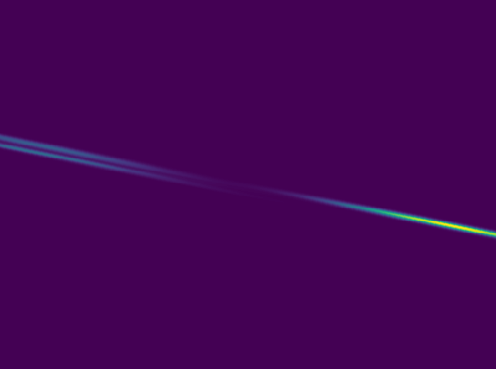};
\end{axis}

\end{tikzpicture}
    \subcaption{$k = 8$}
    \end{subfigure}
    \hfill
    \begin{subfigure}{0.49\linewidth}
\begin{tikzpicture}

\definecolor{darkgray176}{RGB}{176,176,176}

\begin{axis}[
tick align=outside,
tick pos=left,
x grid style={darkgray176},
xlabel={$x^{1}$},
xmin=-2, xmax=2,
xtick style={color=black},
xtick={-2, -1, 0, 1, 2},
xticklabels={
  \(\displaystyle {\ensuremath{-}2}\),
  \(\displaystyle {\ensuremath{-}1}\),
  \(\displaystyle {0}\),
  \(\displaystyle {1}\),
  \(\displaystyle {2}\)
},
y grid style={darkgray176},
ylabel={$x^2$},
ymin=-2, ymax=2,
ytick style={color=black},
ytick={-2, -1, 0, 1, 2},
yticklabels={
  \(\displaystyle {\ensuremath{-}2}\),
  \(\displaystyle {\ensuremath{-}1}\),
  \(\displaystyle {0}\),
  \(\displaystyle {1}\),
  \(\displaystyle {2}\)
},
ylabel style={at={(-0.13,0.5)}},
label style={font=\LARGE}, 
tick label style={font=\Large}
]
\addplot graphics [includegraphics cmd=\pgfimage,xmin=-2, xmax=2, ymin=-2, ymax=2] {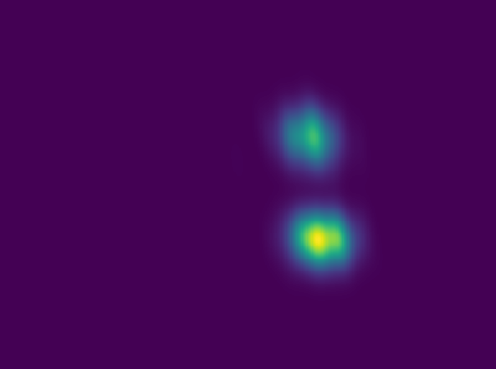};
\end{axis}

\end{tikzpicture}
    \subcaption{$k = 10$}
    \end{subfigure}
    \caption{Density Evolution}
    \label{fig:density-evolution}
\end{figure}

\subsection{Integrator on a Plane with General Densities}\label{ss:num-exp-integrator}
In this numerical experiment, we consider the system dynamics to be a single integrator on the $x-y$ plane with a discrete-time step of $\Delta t = 1.0$. 
The system matrices are defined as follows: $A_k = I_2$, $B_k = \Delta t I_2$, and $x_k, u_k \in \R{2}$ for all $k \in \curly{0, \dots, N}$, where $N = 10$. 
Additionally, the matrices associated with the cost function are specified as follows: $Q_k = \mathbf{0}$, $x_k' = \bm{0}$, and $R_k = I_2'$.

In this example, the initial state distribution is defined as the uniform distribution over the set $\mathcal{S}$, where $\mathcal{S}:= \curly{x, y \in \R{} ~|~ (x, y) \in [0, 5] \times [0, 5]}$. 
On the other hand, the desired state distribution is a uniform distribution over the shape of the letter ``T''. 
Both the initial and desired distributions are visualized in Figure \ref{fig:initial-desired-actual}. 
To represent the densities, we collected 10,000 samples from these distributions and employed kernel density estimation with a Gaussian kernel to generate a continuous probability density function.

Figure \ref{fig:initial-desired-approx} displays the approximated initial and desired state distributions, which were obtained using the expectation-maximization (EM) algorithm. 
In this visualization, we collected 10,000 samples from the approximated Gaussian mixture distributions and applied kernel density estimation, as shown in Figure \ref{fig:initial-desired-actual}. 
Since the EM algorithm iteratively updates the parameters of a GMM to maximize the log-likelihood of the samples, we gathered samples from both the initial and desired distributions. 
For the initial distribution, we collected 2000 samples, and the GMM was configured with $r = 6$ components. 
In contrast, we gathered 2000 samples for the desired distribution, which utilized a GMM with $t = 10$ components. 
For the implementation of EM algorithm and GMM approximations, we used Sckit-learn library \cite{p:scikit-learn}.
It's worth noting that the EM algorithm took 0.36 seconds and 0.08 seconds to converge to locally optimal parameter values. 

Figure \ref{fig:density-evolution-example2} provides a visual representation of the state distribution's evolution over time, specifically showcasing the state distribution at time steps $k = \curly{0, 4, 8, 10}$. 
It's noteworthy that the computation of the optimal policy requires 0.65 seconds. 
The linear program outlined in \eqref{eq:main-lin-prog} involves 50 decision variables, as defined in Section \ref{ss:num-exp-integrator}, allowing for quick solutions within milliseconds. 
The challenge arises in determining the values of the matrix $\mathbf{C}$ in \eqref{eq:main-lin-prog-obj}, which necessitates solving individual mean and covariance steering problems using Proposition \ref{prop:optimal-mean-steering} and Proposition \ref{prop:optimal-covariance-steering}. 
The main computational load stems from repeatedly calculating $(\mathbf{D}\transpose M \mathbf{D})^{-1}$ and $(\mathbf{B}\mathbf{N} M^{-1} \mathbf{B}\mathbf{N}\transpose)^{-1}$. To expedite the computation of $\mathbf{C}$, these matrices can be computed once and stored to avoid redundant calculations.

\begin{figure}
    \centering
    \begin{subfigure}{0.49\linewidth}
\begin{tikzpicture}

\definecolor{darkgray176}{RGB}{176,176,176}

\begin{axis}[
tick align=outside,
tick pos=left,
xlabel={$x$},
ylabel={$y$},
x grid style={darkgray176},
xmin=-0.515075376884422, xmax=5.51507537688442,
xtick style={color=black},
xtick={-1,0,1,2,3,4,5,6},
xticklabels={
  \(\displaystyle {\ensuremath{-}1}\),
  \(\displaystyle {0}\),
  \(\displaystyle {1}\),
  \(\displaystyle {2}\),
  \(\displaystyle {3}\),
  \(\displaystyle {4}\),
  \(\displaystyle {5}\),
  \(\displaystyle {6}\)
},
y grid style={darkgray176},
ymin=-0.515075376884422, ymax=5.51507537688442,
ytick style={color=black},
ytick={-1,0,1,2,3,4,5,6},
yticklabels={
  \(\displaystyle {\ensuremath{-}1}\),
  \(\displaystyle {0}\),
  \(\displaystyle {1}\),
  \(\displaystyle {2}\),
  \(\displaystyle {3}\),
  \(\displaystyle {4}\),
  \(\displaystyle {5}\),
  \(\displaystyle {6}\)
},
label style={font=\LARGE}, 
tick label style={font=\Large}
]
\addplot graphics [includegraphics cmd=\pgfimage,xmin=-0.515075376884422, xmax=5.51507537688442, ymin=-0.515075376884422, ymax=5.51507537688442] {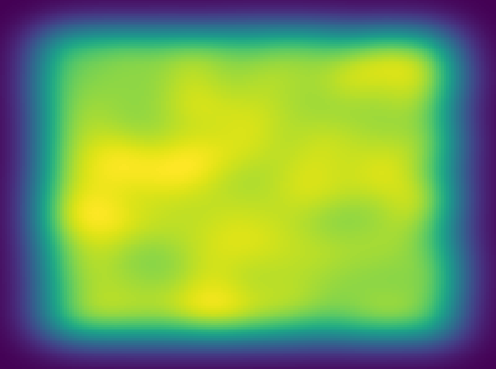};
\end{axis}

\end{tikzpicture}
    \subcaption{Initial Density}
    \end{subfigure}
    \hfill
    \begin{subfigure}{0.49\linewidth}
\begin{tikzpicture}

\definecolor{darkgray176}{RGB}{176,176,176}

\begin{axis}[
tick align=outside,
tick pos=left,
xlabel={$x$},
ylabel={$y$},
x grid style={darkgray176},
xmin=-0.507518796992481, xmax=5.50751879699248,
xtick style={color=black},
xtick={-1,0,1,2,3,4,5,6},
xticklabels={
  \(\displaystyle {\ensuremath{-}1}\),
  \(\displaystyle {0}\),
  \(\displaystyle {1}\),
  \(\displaystyle {2}\),
  \(\displaystyle {3}\),
  \(\displaystyle {4}\),
  \(\displaystyle {5}\),
  \(\displaystyle {6}\)
},
y grid style={darkgray176},
ymin=-0.507518796992481, ymax=5.50751879699248,
ytick style={color=black},
ytick={-1,0,1,2,3,4,5,6},
yticklabels={
  \(\displaystyle {\ensuremath{-}1}\),
  \(\displaystyle {0}\),
  \(\displaystyle {1}\),
  \(\displaystyle {2}\),
  \(\displaystyle {3}\),
  \(\displaystyle {4}\),
  \(\displaystyle {5}\),
  \(\displaystyle {6}\)
},
label style={font=\LARGE}, 
tick label style={font=\Large}
]
\addplot graphics [includegraphics cmd=\pgfimage,xmin=-0.507518796992481, xmax=5.50751879699248, ymin=-0.507518796992481, ymax=5.50751879699248] {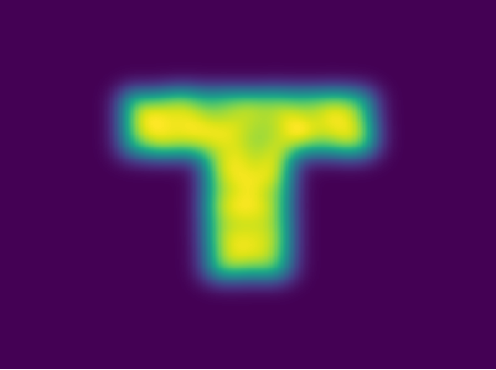};
\end{axis}

\end{tikzpicture}
    \subcaption{Desired Density}
    \end{subfigure}    
    \caption{Initial and Desired Densities (Actual)}
    \label{fig:initial-desired-actual}
\end{figure}

\begin{figure}
    \centering
    \begin{subfigure}{0.49\linewidth}
\begin{tikzpicture}

\definecolor{darkgray176}{RGB}{176,176,176}

\begin{axis}[
tick align=outside,
tick pos=left,
xlabel={$x$},
ylabel={$y$},
x grid style={darkgray176},
xmin=-0.515075376884422, xmax=5.51507537688442,
xtick style={color=black},
xtick={-1,0,1,2,3,4,5,6},
xticklabels={
  \(\displaystyle {\ensuremath{-}1}\),
  \(\displaystyle {0}\),
  \(\displaystyle {1}\),
  \(\displaystyle {2}\),
  \(\displaystyle {3}\),
  \(\displaystyle {4}\),
  \(\displaystyle {5}\),
  \(\displaystyle {6}\)
},
y grid style={darkgray176},
ymin=-0.515075376884422, ymax=5.51507537688442,
ytick style={color=black},
ytick={-1,0,1,2,3,4,5,6},
yticklabels={
  \(\displaystyle {\ensuremath{-}1}\),
  \(\displaystyle {0}\),
  \(\displaystyle {1}\),
  \(\displaystyle {2}\),
  \(\displaystyle {3}\),
  \(\displaystyle {4}\),
  \(\displaystyle {5}\),
  \(\displaystyle {6}\)
},
label style={font=\LARGE}, 
tick label style={font=\Large}
]
\addplot graphics [includegraphics cmd=\pgfimage,xmin=-0.515075376884422, xmax=5.51507537688442, ymin=-0.515075376884422, ymax=5.51507537688442] {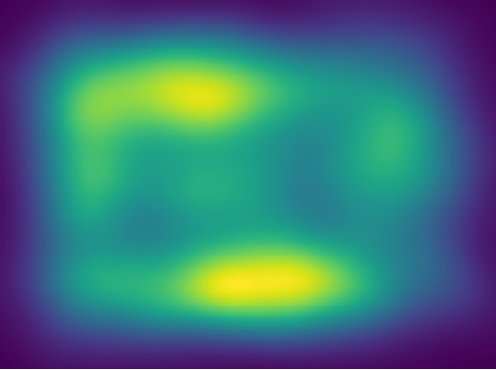};
\end{axis}

\end{tikzpicture}
    \subcaption{Initial Density}
    \end{subfigure}
    \hfill
    \begin{subfigure}{0.49\linewidth}
\begin{tikzpicture}

\definecolor{darkgray176}{RGB}{176,176,176}

\begin{axis}[
tick align=outside,
tick pos=left,
xlabel={$x$},
ylabel={$y$},
x grid style={darkgray176},
xmin=-0.515075376884422, xmax=5.51507537688442,
xtick style={color=black},
xtick={-1,0,1,2,3,4,5,6},
xticklabels={
  \(\displaystyle {\ensuremath{-}1}\),
  \(\displaystyle {0}\),
  \(\displaystyle {1}\),
  \(\displaystyle {2}\),
  \(\displaystyle {3}\),
  \(\displaystyle {4}\),
  \(\displaystyle {5}\),
  \(\displaystyle {6}\)
},
y grid style={darkgray176},
ymin=-0.515075376884422, ymax=5.51507537688442,
ytick style={color=black},
ytick={-1,0,1,2,3,4,5,6},
yticklabels={
  \(\displaystyle {\ensuremath{-}1}\),
  \(\displaystyle {0}\),
  \(\displaystyle {1}\),
  \(\displaystyle {2}\),
  \(\displaystyle {3}\),
  \(\displaystyle {4}\),
  \(\displaystyle {5}\),
  \(\displaystyle {6}\)
},
label style={font=\LARGE}, 
tick label style={font=\Large}
]
\addplot graphics [includegraphics cmd=\pgfimage,xmin=-0.515075376884422, xmax=5.51507537688442, ymin=-0.515075376884422, ymax=5.51507537688442] {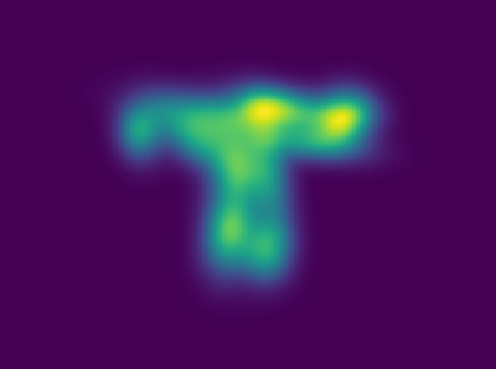};
\end{axis}

\end{tikzpicture}
    \subcaption{Desired Density}
    \end{subfigure}
    \caption{Initial and Desired Densities (Approximated)}
    \label{fig:initial-desired-approx}
\end{figure}

\begin{figure*}[ht]
    \centering
    \begin{subfigure}{0.24\linewidth}
\begin{tikzpicture}

\definecolor{darkgray176}{RGB}{176,176,176}

\begin{axis}[
tick align=outside,
tick pos=left,
x grid style={darkgray176},
xlabel={$x$},
ylabel={$y$},
xmin=-0.5, xmax=5.5,
xtick style={color=black},
xtick={-1,0,1,2,3,4,5,6},
xticklabels={
  \(\displaystyle {\ensuremath{-}1}\),
  \(\displaystyle {0}\),
  \(\displaystyle {1}\),
  \(\displaystyle {2}\),
  \(\displaystyle {3}\),
  \(\displaystyle {4}\),
  \(\displaystyle {5}\),
  \(\displaystyle {6}\)
},
y grid style={darkgray176},
ymin=-0.5, ymax=5.5,
ytick style={color=black},
ytick={-1,0,1,2,3,4,5,6},
yticklabels={
  \(\displaystyle {\ensuremath{-}1}\),
  \(\displaystyle {0}\),
  \(\displaystyle {1}\),
  \(\displaystyle {2}\),
  \(\displaystyle {3}\),
  \(\displaystyle {4}\),
  \(\displaystyle {5}\),
  \(\displaystyle {6}\)
},
label style={font=\LARGE}, 
tick label style={font=\Large}
]
\addplot graphics [includegraphics cmd=\pgfimage,xmin=-0.5, xmax=5.5, ymin=-0.5, ymax=5.5] {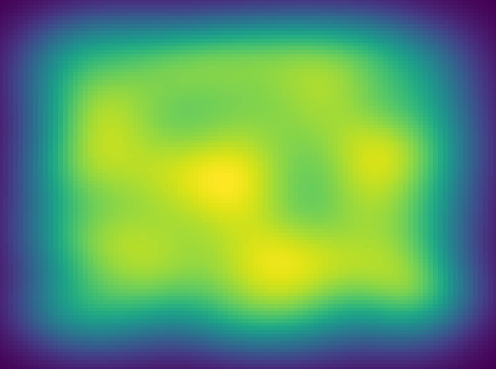};
\end{axis}

\end{tikzpicture}
    \subcaption{$k=0$}
    \end{subfigure}
    \begin{subfigure}{0.24\linewidth}
\begin{tikzpicture}

\definecolor{darkgray176}{RGB}{176,176,176}

\begin{axis}[
tick align=outside,
tick pos=left,
x grid style={darkgray176},
xlabel={$x$},
ylabel={$y$},
xmin=-0.5, xmax=5.5,
xtick style={color=black},
xtick={-1,0,1,2,3,4,5,6},
xticklabels={
  \(\displaystyle {\ensuremath{-}1}\),
  \(\displaystyle {0}\),
  \(\displaystyle {1}\),
  \(\displaystyle {2}\),
  \(\displaystyle {3}\),
  \(\displaystyle {4}\),
  \(\displaystyle {5}\),
  \(\displaystyle {6}\)
},
y grid style={darkgray176},
ymin=-0.5, ymax=5.5,
ytick style={color=black},
ytick={-1,0,1,2,3,4,5,6},
yticklabels={
  \(\displaystyle {\ensuremath{-}1}\),
  \(\displaystyle {0}\),
  \(\displaystyle {1}\),
  \(\displaystyle {2}\),
  \(\displaystyle {3}\),
  \(\displaystyle {4}\),
  \(\displaystyle {5}\),
  \(\displaystyle {6}\)
},
label style={font=\LARGE}, 
tick label style={font=\Large}
]
\addplot graphics [includegraphics cmd=\pgfimage,xmin=-0.5, xmax=5.5, ymin=-0.5, ymax=5.5] {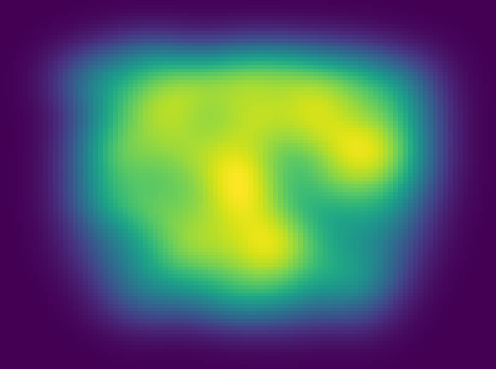};
\end{axis}

\end{tikzpicture}
    \subcaption{$k=4$}
    \end{subfigure}
    \begin{subfigure}{0.24\linewidth}
\begin{tikzpicture}

\definecolor{darkgray176}{RGB}{176,176,176}

\begin{axis}[
tick align=outside,
tick pos=left,
x grid style={darkgray176},
xlabel={$x$},
ylabel={$y$},
xmin=-0.5, xmax=5.5,
xtick style={color=black},
xtick={-1,0,1,2,3,4,5,6},
xticklabels={
  \(\displaystyle {\ensuremath{-}1}\),
  \(\displaystyle {0}\),
  \(\displaystyle {1}\),
  \(\displaystyle {2}\),
  \(\displaystyle {3}\),
  \(\displaystyle {4}\),
  \(\displaystyle {5}\),
  \(\displaystyle {6}\)
},
y grid style={darkgray176},
ymin=-0.5, ymax=5.5,
ytick style={color=black},
ytick={-1,0,1,2,3,4,5,6},
yticklabels={
  \(\displaystyle {\ensuremath{-}1}\),
  \(\displaystyle {0}\),
  \(\displaystyle {1}\),
  \(\displaystyle {2}\),
  \(\displaystyle {3}\),
  \(\displaystyle {4}\),
  \(\displaystyle {5}\),
  \(\displaystyle {6}\)
},
label style={font=\LARGE}, 
tick label style={font=\Large}
]
\addplot graphics [includegraphics cmd=\pgfimage,xmin=-0.5, xmax=5.5, ymin=-0.5, ymax=5.5] {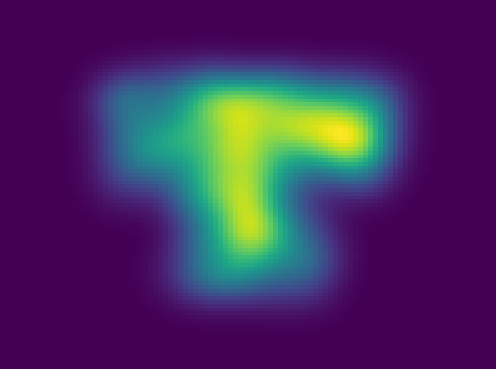};
\end{axis}

\end{tikzpicture}
    \subcaption{$k=8$}
    \end{subfigure}
    \begin{subfigure}{0.24\linewidth}
\begin{tikzpicture}

\definecolor{darkgray176}{RGB}{176,176,176}

\begin{axis}[
tick align=outside,
tick pos=left,
x grid style={darkgray176},
xlabel={$x$},
ylabel={$y$},
xmin=-0.5, xmax=5.5,
xtick style={color=black},
xtick={-1,0,1,2,3,4,5,6},
xticklabels={
  \(\displaystyle {\ensuremath{-}1}\),
  \(\displaystyle {0}\),
  \(\displaystyle {1}\),
  \(\displaystyle {2}\),
  \(\displaystyle {3}\),
  \(\displaystyle {4}\),
  \(\displaystyle {5}\),
  \(\displaystyle {6}\)
},
y grid style={darkgray176},
ymin=-0.5, ymax=5.5,
ytick style={color=black},
ytick={-1,0,1,2,3,4,5,6},
yticklabels={
  \(\displaystyle {\ensuremath{-}1}\),
  \(\displaystyle {0}\),
  \(\displaystyle {1}\),
  \(\displaystyle {2}\),
  \(\displaystyle {3}\),
  \(\displaystyle {4}\),
  \(\displaystyle {5}\),
  \(\displaystyle {6}\)
},
label style={font=\LARGE}, 
tick label style={font=\Large}
]
\addplot graphics [includegraphics cmd=\pgfimage,xmin=-0.5, xmax=5.5, ymin=-0.5, ymax=5.5] {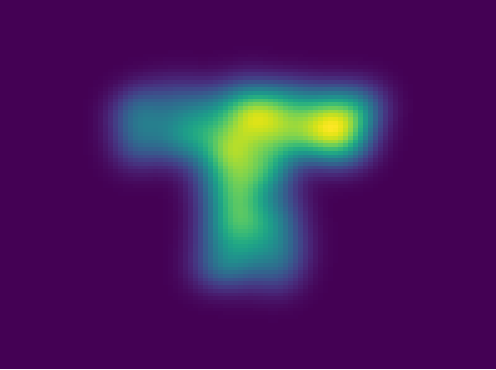};
\end{axis}

\end{tikzpicture}
    \subcaption{$k=10$}
    \end{subfigure}
    \caption{Density Evolution}
    \label{fig:density-evolution-example2}
\end{figure*}

\section{Conclusion}\label{s:conclusion} 
In this paper, we studied the optimal density steering problem for discrete-time linear dynamical systems whose initial and  target state distributions can be represented as Gaussian mixture models. We proposed a finite-dimensional class of control policies that can be used to reduce the density steering problem into an equivalent nonlinear program which we subsequently reduced into a tractable linear program. Finally, we demonstrated the efficacy of our approach in non-trivial numerical simulations. In our future work, we will address constrained density steering problems based on Gaussian mixture models, by utilizing the divergence metrics between probability distributions represented as Gaussian mixture models \cite{p:kampa2011ClosedFormCSDivergenceGMM, p:chen2018OT-GMM}.

\bibliographystyle{ieeetr}
\bibliography{density_steering}

\end{document}